\documentclass[a4paper,11pt]{article}
\pdfoutput=1
\usepackage{jheppub}
\usepackage[T1]{fontenc}

\usepackage{subcaption}
\usepackage{amsthm}
\usepackage{array}
\newtheorem{prop}{Proposition}[section]

\title{Narain CFTs and error-correcting codes on finite fields}

\author{Shinichiro Yahagi}
\affiliation{Department of Physics, The University of Tokyo, \\
7-3-1, Hongo, Bunkyo-ku, Tokyo 113-0033, Japan}
\emailAdd{yahagi@hep-th.phys.s.u-tokyo.ac.jp}

\abstract{We construct Narain CFTs from self-dual codes on the finite field $F_p$ through even self-dual lattices for any prime $p>2$. Using this correspondence, we can relate the spectral gap and the partition function of the CFT to the error correction capability and the extended enumerator polynomial of the code. In particular, we calculate specific spectral gaps of CFTs constructed from codes and compare them with the largest spectral gap among all Narain CFTs.}

\begin{document}
\maketitle
\flushbottom

\section{Introduction}
In high energy physics, conformal field theories (CFTs) are widely used such as for describing string theory. On the other hand, in information theory, error-correcting codes enable us to transmit information correctly in spite of possible errors. Although their theoretical motivations are very different, their properties have many mathematical similarities.

In fact, the relation between classical error-correcting codes on $F_2=\{0,1\}$ and chiral CFTs through Euclidean lattices has long been known. The symmetries and quantities of CFTs have been analyzed through codes \cite{crl}, and vertex operator algebras have been studied such as the correspondence among the Golay code, the CFT related to the Leech lattice and the Monster group \cite{VOA, 1004.0956}. Recently, the higher genus partition functions of CFTs constructed from codes were calculated to understand modular invariance at higher genus \cite{2112.05168}.

There is also a rich connection between quantum error-correcting codes and CFTs. The term ``quantum'' means transmitting a quantum state that can involve quantum entanglement among qubits. For example, the symmetries of quantum codes can be translated into the superconformal symmetries of corresponding CFTs \cite{2003.13700}. In addition, the properties of holography such as the AdS/CFT correspondence can be interpreted in the language of quantum codes and tensor networks \cite{1411.7041, 2102.02619}.

As discussed later, we focus on the spectrum in this paper. The spectrum of two-dimensional CFTs has been studied numerically and analytically using the modular invariance of CFTs and linear functional methods, which is related to the sphere packing problem \cite{1307.6562, 1608.06241, 1905.01319, 2006.02560}. In particular, the maximum value of the spectral gaps that Narain CFTs can have was evaluated numerically using the modular bootstrap, and analytical methods to get the lower bound are proposed \cite{2006.04839}. The spectrum itself is an essential quantity for the CFT, and it is also related to states in the corresponding theory of gravity through holography \cite{2012.15830}.

Recently, the relation between classical error-correcting codes on $F_2$ or $F_4=\{0,1,\omega,\omega^2\}$ and Narain CFTs through lattices on $\mathbb{R}^{n,n}$ have been studied and quantities in CFTs such as spectral gaps and partition functions were described in the language of codes \cite{2009.01244,2107.02816}. In this paper, we extend the relation to codes on $F_p$ for any prime $p>2$ and more general Narain CFTs. The relation between CFTs, lattices and codes is roughly summarized in Table \ref{tab:relation}. In addition, we calculate spectral gaps of corresponding Narain CFTs for small $p$ and compare them with known optimal values in \cite{2006.04839}. Though it is generally difficult to solve the problem of finding the spectral gap for a given CFT, we can calculate it rigorously using the construction from codes.

\begin{table}[tbp]
	\caption{The relation between CFTs, lattices and codes}
	\label{tab:relation}
	\centering
	\begin{tabular}{c||wc{0.25\textwidth}|wc{0.25\textwidth}|wc{0.25\textwidth}}
	& CFT & Lattice $\Lambda$ & Code $\mathcal{C}$ \\	\hline \hline
	& modular invariance & even self-dual & self-dual \\ \hline
	$n$ & central charge & dimension & length \\ \hline
	$p$ & \begin{tabular}{c} compactification\\radii $\sqrt{2/p}$ \end{tabular} & $\sqrt{p}\mathbb{Z}^{2n}\subset\Lambda$ & \begin{tabular}{c} on the finite field\\with $p$ elements \end{tabular} \\ \hline
	& spectral gap & minimum length & correction capability \\ \cline{2-4}
	& partition function & & enumerator polynomial \\ \hline
	\end{tabular}
\end{table}

The outline of the paper is as follows. In section \ref{sec:pre}, we review the background knowledge about Narain CFTs, lattices and codes. We clarify the properties ``even'' and ``self-dual'', which play important roles in establishing the relation.
Section \ref{sec:relation} contains the main result of this paper. In subsection \ref{sec:construction_flat}, we construct Narain CFTs from self-dual codes on $F_p^{2n}$ through even self-dual lattices on $\mathbb{R}^{n,n}$ for any prime $p>2$. The corresponding CFTs have $n$ scalars, radii $\sqrt{2/p}$, the flat metric and antisymmetric tensor backgrounds with integer values (hereinafter called integer backgrounds). In subsection \ref{sec:quantities}, we relate the spectral gap and the partition function in the CFT to the error correction capability and the enumerator polynomial in the corresponding code. Using this relation, we calculate specific spectral gaps of Narain CFTs constructed from codes for small $p,n$.
In subsection \ref{sec:construction_general}, we introduce a more general construction that covers Narain CFTs with curved metrics and non-integer backgrounds. Since the relation between quantities is not concise in this case, we give a few important examples.

\section{Preliminaries} \label{sec:pre}
This section deals with the background knowledge about Narain CFTs, lattices and codes needed for later discussions to understand the relation in Table \ref{tab:relation}.

\subsection{Narain CFTs}
A Narain CFT is described by $n$ scalars $X^i,i=1,\dots,n$ on a torus $T^n$ with radii $R$ \cite{narain, 9401139}. We follow the notations of \cite{polchinski}. A world-sheet action $S$ is
\begin{equation}
	S = \frac{1}{4\pi\alpha'} \int dt \int_0^{2\pi} d\sigma \left[ G_{ij} (\partial_tX^i\partial_tX^j-\partial_\sigma X^i\partial_\sigma X^j) - 2B_{ij} \partial_tX^i\partial_\sigma X^j \right]
\end{equation}
where $G_{ij}$ is a metric (symmetric) and $B_{ij}$ is an antisymmetric tensor background, both of which are constants.
From the equation of motion, the mode expansion of $X^i$ is
\begin{align}
	X^i(t,\sigma) &= X_L^i(t-\sigma) + X_R^i(t+\sigma), \label{eq:X} \\
	X_L^i(t-\sigma) &= x_L^i + \frac{\alpha'}{2}p_L^i(t-\sigma) + i\sqrt{\frac{\alpha'}{2}}\sum_{n\in\mathbb{Z}\setminus\{0\}}\frac{\alpha_n^i}{n}e^{-in(t-\sigma)}, \label{eq:X_L} \\
	X_R^i(t+\sigma) &= x_R^i + \frac{\alpha'}{2}p_R^i(t+\sigma) + i\sqrt{\frac{\alpha'}{2}}\sum_{n\in\mathbb{Z}\setminus\{0\}}\frac{\tilde{\alpha}_n^i}{n}e^{-in(t+\sigma)}. \label{eq:X_R}
\end{align}
From the compactification, they should satisfy
\begin{gather}
	X^i(t,\sigma) - X^i(t,\sigma+2\pi) = 2\pi R w^i, \label{eq:wind} \\
	P_i := \frac{\partial L}{\partial (\partial_t X^i)} = \frac{1}{R}m_i \label{eq:momentum}
\end{gather}
where $\vec{w}=(w^1,\dots,w^n)^T,\vec{m}=(m_1,\dots,m_n)^T\in\mathbb{Z}^n$ and $S=\int dt L$. By substituting (\ref{eq:X}),(\ref{eq:X_L}),(\ref{eq:X_R}) for (\ref{eq:wind}),(\ref{eq:momentum}), the eigenvalues $k_L,k_R$ of $p_L,p_R$ are
\begin{align}
	k_{Li} &= \frac{1}{R} m_i + \frac{R}{\alpha'}(B+G)_{ij}w^j, \\
	k_{Ri} &= \frac{1}{R} m_i + \frac{R}{\alpha'}(B-G)_{ij}w^j. 
\end{align}
We introduce a tetrad $e_i^\mu$ which satisfy $G_{ij}=e_i^\mu e_j^\nu \delta_{\mu\nu}$, its inverse $e_\mu^i$ and dimensionless momenta on orthogonal basis
\begin{align}
	l_{L\mu} &= e_\mu^i \sqrt{\frac{\alpha'}{2}} k_{Li} = e_\mu^i \left[ \frac{1}{r} m_i + \frac{r}{2} (B+G)_{ij}w^j \right], \label{eq:l_L} \\
	l_{R\mu} &= e_\mu^i \sqrt{\frac{\alpha'}{2}} k_{Ri} = e_\mu^i \left[ \frac{1}{r} m_i + \frac{r}{2} (B-G)_{ij}w^j \right] \label{eq:l_R}
\end{align}
where $r=R\sqrt{2/\alpha'}$ is a dimensionless radius. In the following, we use the momentum to refer to this dimensionless one.

The spectral gap and the partition function are important quantities in the Narain CFT. The spectral gap $\Delta$ of primary states, which is the energy difference between the ground state and the first nontrivial primary state, is
\begin{equation}
	\Delta = \min_{\substack{\vec{m},\vec{w}\in\mathbb{Z}^n \\ (\vec{m},\vec{w})\neq\vec{0}}} \frac{\vec{l}_L^2+\vec{l}_R^2}{2}
\end{equation}
when we set $\alpha'=2$. The partition function $Z(\tau)$, which is a trace of states weighted by the Hamiltonian and the momentum, is
\begin{equation}
	Z(\tau) = |\eta(\tau)|^{-2n} \sum_{\vec{m},\vec{w}\in\mathbb{Z}^n} q^{\vec{l}_L^2/2} \bar{q}^{\vec{l}_R^2/2}
\end{equation}
where $\eta(\tau)$ is the Dedekind eta function:
\begin{equation}
	\eta(\tau)=q^{1/24}\prod_{m=1}^\infty(1-q^m),\quad q=e^{2\pi i\tau},\quad \bar{q}=e^{-2\pi i\bar{\tau}}.
\end{equation}

\subsection{Lattices}
A lattice $\Lambda\subset\mathbb{R}^n$ is a subgroup of $\mathbb{R}^n$ which can be written as $\Lambda=\{\sum_{i=1}^n a_iv_i \mid a_i\in\mathbb{Z}\}$ where $v$ are basis of $\mathbb{R}^n$.  Much of the discussion in this and the next subsection depends on \cite{conway}. Since the properties ``even'' and ``self-dual'' are the core of our discussion, we clarify definitions.

\begin{itemize}
	\item ``even'': A lattice $\Lambda\subset\mathbb{R}^n$ is even with a metric $g$ if
	\begin{equation}
		\forall x\in\Lambda,\, x \cdot x \in 2\mathbb{Z}.
	\end{equation}
	\item ``self-dual'': For a lattice $\Lambda\subset\mathbb{R}^n$, a dual lattice with a metric $g$ is
	\begin{equation}
		\Lambda^\ast = \left\{ x'\in\mathbb{R}^n \,\middle|\, \forall x\in\Lambda, x\cdot x'\in\mathbb{Z} \right\} \subset \mathbb{R}^n.
	\end{equation}
	$\Lambda$ is self-dual if $\Lambda = \Lambda^\ast$.
\end{itemize}
Here, the dot ``$\cdot$'' means an inner product $:\mathbb{R}^n\times \mathbb{R}^n\to \mathbb{R}$ defined by $x\cdot y=\sum_{i,j=1}^n g_{ij} x_iy_j$.

It is convenient to represent lattices by matrices. An $n\times n$ matrix $G\in GL(n,\mathbb{R})$ is called a generator matrix of a lattice $\Lambda$ if
\begin{equation}
		\Lambda = \{ Gm \mid m\in\mathbb{Z}^n \}.
\end{equation}
In addition, $G, G'\in GL(n,\mathbb{R})$ are generator matrices of the same lattice if and only if there exists an $n\times n$ unimodular matrix $U$ s.t. $G'=GU$.

The momenta of the Narain CFT form a lattice
\begin{equation}
	\Lambda(r,G,B) = \left\{ \begin{pmatrix} \vec{l}_L \\ \vec{l}_R \end{pmatrix} \,\middle|\, \vec{m},\vec{w}\in\mathbb{Z}^n \right\} \subset \mathbb{R}^{2n}
\end{equation}
and we can easily prove that this lattice is even self-dual with a metric
\begin{equation}
	g' = \begin{pmatrix} I&0 \\ 0&-I \end{pmatrix}
\end{equation}
where $I$ is an $n\times n$ identity matrix. We can define another lattice by
\begin{equation}
	\Lambda_N(r,G,B) = \left\{ \begin{pmatrix} \alpha \\ \beta \end{pmatrix} \,\middle|\, \vec{m},\vec{w}\in\mathbb{Z}^n \right\} \subset \mathbb{R}^{2n}
\end{equation}
where
\begin{equation}
	\alpha = \frac{\vec{l}_L+\vec{l}_R}{\sqrt{2}}, \, 
	\beta = \frac{\vec{l}_L-\vec{l}_R}{\sqrt{2}}.
\end{equation}
This lattice is even self-dual with a metric
\begin{equation}
	g = \begin{pmatrix} 0&I \\ I&0 \end{pmatrix}
\end{equation}
since the inner product with $g$ of vectors in $\Lambda_N(r,G,B)$ is equal to the inner product with $g'$ of the corresponding vectors in $\Lambda(r,G,B)$. From \eqref{eq:l_L} and \eqref{eq:l_R}, $\alpha$ and $\beta$ can be written as
\begin{equation}
	\alpha_\mu = e_\mu^i \left( \frac{\sqrt{2}}{r} m_i + \frac{r}{\sqrt{2}} B_{ij}w^j \right), \  
	\beta_\mu = e_\mu^i \frac{r}{\sqrt{2}} G_{ij}w^j,
\end{equation}
and thus one generator matrix of $\Lambda_N(r,G,B)$ is
\begin{equation} \label{eq:lambda_N_gen_mat}
	\begin{pmatrix}
			\frac{\sqrt{2}}{r}\gamma^{-1} & \frac{r}{\sqrt{2}}\gamma^{-1}B \\
			0 & \frac{r}{\sqrt{2}} \gamma^T
	\end{pmatrix}
\end{equation}
where $(\gamma)_{i\mu} = e_i^\mu$. We will associate a code with this lattice in the next section. 

In the language of the lattice, the spectral gap $\Delta$ is
\begin{equation} \label{eq:spectral_gap}
	\Delta =  \min_{\substack{(\vec{l}_L,\vec{l}_R)\in\Lambda(r,G,B) \\ \vec{l}_L^2+\vec{l}_R^2\neq0}} \frac{\vec{l}_L^2+\vec{l}_R^2}{2}
	= \min_{\substack{(\alpha,\beta)\in\Lambda_N(r,G,B) \\ \alpha^2+\beta^2\neq0}} \frac{\alpha^2+\beta^2}{2}
\end{equation}
and the partition function $Z(\tau)$ is
\begin{align}
	Z(\tau) &= |\eta(\tau)|^{-2n} \sum_{(\vec{l}_L,\vec{l}_R)\in\Lambda(r,G,B)} q^{\vec{l}_L^2/2} \bar{q}^{\vec{l}_R^2/2} \\
	&= |\eta(\tau)|^{-2n} \sum_{(\alpha,\beta)\in\Lambda_N(r,G,B)} q^{(\alpha+\beta)^2/4} \bar{q}^{(\alpha-\beta)^2/4}. \label{eq:partition_function}
\end{align}
The partition function of the Narain CFT is modular-invariant, i.e. $Z(\tau)=Z(\tau+1)$ and $Z(\tau)=Z(-1/\tau)$ which respectively correspond to evenness and self-duality of the lattice ($\Lambda(r,G,B)$ with the metric $g'$ or $\Lambda_N(r,G,B)$ with the metric $g$).

\subsection{Codes}
For a prime $p$, we will consider an error-correcting code on $F_p=\mathbb{Z}/p\mathbb{Z}$, which is a finite field with $p$ elements. In this paper, elements on $F_p$ are often denoted by bars, for instance, $F_3=\{\bar{0},\bar{1},\bar{2}\}$, in order to distinguish them from elements on $\mathbb{R}$. Furthermore, $R:F_p\to \mathbb{Z}$ is a map s.t. $R(\bar{x})=x$. We define a distance between $a,b\in F_p^n$ by
\begin{equation}
	d(a,b) = \sqrt{ \sum_{i=1}^n (\min\{R(a_i-b_i), R(b_i-a_i)\})^2 }
\end{equation}
and write $d(a,0)=d(a)$. It can be interpreted as an analog of the Euclidean distance on a torus $T^n$ with length $p$.

A code on $F_p$ can be expressed by a subspace $\mathcal{C}\subset F_p^n$ where $n$ is called length. In particular, a code $\mathcal{C}$ is additive if $a+b\in\mathcal{C}$ for all $a,b\in\mathcal{C}$. For an additive code $\mathcal{C}\subset F_p^n$, there exists an integer $k$ s.t. $\mathcal{C}$ has $p^k$ elements, and an $n\times k$ matrix $G$ is called a generator matrix if
\begin{equation}
	\mathcal{C} = \{ Ga \mid a\in F_p^k \}.
\end{equation}

The code is used to transmit information. During the transmission, it may be possible to pick errors. Now, we assume that an error $e\in F_p^n: a\to a+e$ occurs with a probability inversely proportional to $d(e)$. In this case, the coding procedure through $\mathcal{C}$ with a specified $n\times k$ generator matrix $G$ on $F_p$ is as follows:
\begin{enumerate}
	\item Alice wants to transmit $a\in F_p^k$ to Bob.
	\item Alice sends $Ga\in \mathcal{C}\subset F_p^n$ to Bob.
	\item Bob receives $b=Ga+e\in F_p^n$. (If no error occurred, $b=Ga$.)
	\item Bob decodes $b$ by $\mathrm{arg}\min_{c \in\mathcal{C}} d(b,c)$.
\end{enumerate}
Bob can get the correct message $a$ if $d(b,Ga)<d(b,c)$ for all $c\neq Ga\in\mathcal{C}$. Thus, we can always correct the error $e$ if $d(e)<D(\mathcal{C})/2$ where $D(\mathcal{C})$ is a distance of the code defined by
\begin{equation}
	D(\mathcal{C}) = \min_{a,b\in\mathcal{C},a\neq b} d(a,b).
\end{equation}
The distance does not depend on the choice of the generator matrix and represents the error correction capability of the code. If $\mathcal{C}$ is additive,
\begin{equation}
	D(\mathcal{C}) = \min_{c\in\mathcal{C},c\neq0} d(c)
\end{equation}
since $d(a,b)=d(a-b,0)$ and $a-b\in\mathcal{C}$. Note that $D(\mathcal{C})$ is not the Hamming distance.

As for lattices, we define the property ``self-dual'' for codes.
\begin{itemize}
	\item ``self-dual'': For a code $\mathcal{C}\subset F_p^n$, a dual code with a metric $\bar{g}$ on $F_p^n$ is
	\begin{equation}
		\mathcal{C}^\ast = \{ c' \in F_p^n \mid \forall c\in\mathcal{C}, c\circ c'= \bar{0} \} \subset F_p^n.
	\end{equation}
	$\mathcal{C}$ is self-dual if $\mathcal{C} = \mathcal{C}^\ast$.
\end{itemize}
Here, the circle ``$\circ$'' means an inner product $:F_p^n\times F_p^n\to F_p$ defined by $a\circ b=\sum_{i,j=1}^n \bar{g}_{ij}a_ib_j$.

\section{Narain CFTs and codes} \label{sec:relation}
This section is the main part of this paper. We will define the correspondence between codes on $F_p^{2n}$ and Narain CFTs with $n$ scalars and see the relation between quantities in both theories.

In this section, we set
\begin{equation}
	g = \begin{pmatrix} 0&I \\ I&0 \end{pmatrix},\quad
	\bar{g} = \begin{pmatrix} \bar{0}&\bar{I} \\ \bar{I}&\bar{0} \end{pmatrix}
\end{equation}
where $I,\bar{I}$ are $n\times n$ identity matrices on $\mathbb{R},F_p$. Furthermore, ``$\cdot$'' and ``$\circ$'' mean inner products respectively on $\mathbb{R}^{2n}$ with the metric $g$ and on $F_p^{2n}$ with the metric $\bar{g}$.

\subsection{Construction of Narain CFTs from codes} \label{sec:construction_flat}
We define a lattice corresponding to an additive code $\mathcal{C}\subset F_p^{2n}$ by
\begin{equation} \label{eq:lambda_p_C}
	\Lambda_p(\mathcal{C}) = \left\{ \frac{R(c)+pm}{\sqrt{p}} \,\middle|\, c\in\mathcal{C}, m\in \mathbb{Z}^{2n} \right\} \subset \frac{1}{\sqrt{p}} \mathbb{Z}^{2n} \subset \mathbb{R}^{2n}.
\end{equation}
It is obvious from the definition that $\Lambda_p(\mathcal{C})=\Lambda_p(\mathcal{C}')$ if and only if $\mathcal{C}=\mathcal{C}'$. In this construction, we can translate the properties ``even'' and ``self-dual'' into words in codes as follows.
\begin{prop} \label{prop:lambda_p_even}
	For a prime $p\neq2$, $\Lambda_p(\mathcal{C})$ is even with the metric $g$ if and only if $c\circ c=\bar{0}$ for all $c\in\mathcal{C}$.
\end{prop}
\begin{proof}
	We rewrite the evenness of the lattice. From the definition, $\Lambda_p(\mathcal{C})$ is even with $g$ when $x\cdot x\in 2\mathbb{Z}$ for all $x\in\Lambda_p(\mathcal{C})$. From the construction \eqref{eq:lambda_p_C}, it can be rewritten that
	\begin{equation}
	\begin{aligned}
		& \frac{1}{\sqrt{p}} (R(c)+pm)\cdot \frac{1}{\sqrt{p}} (R(c)+pm) \\
		=\;& \frac{1}{p} R(c)\cdot R(c) + 2R(c)\cdot m + p m\cdot m \\
		=\;& \frac{2}{p} \sum_{i=1}^n R(c_i)R(c_{i+n}) + 2\sum_{i=1}^n (R(c_i)m_{i+n}+R(c_{i+n})m_i) + 2p \sum_{i=1}^n m_im_{i+n}
	\end{aligned}
	\end{equation}
	is even for all $c\in\mathcal{C}$ and $m\in\mathbb{Z}^{2n}$. Since the second and third terms are always even, it is satisfied if and only if $\sum_{i=1}^n R(c_i)R(c_{i+n})\in p\mathbb{Z}$. The calculation on $F_p$ can be regarded as that on $\mathbb{R}$ with mod $p$, thus it can be expressed as $\sum_{i=1}^nc_ic_{i+n}=\bar{0}$. From $c\circ c=2\sum_{i=1}^n c_ic_{i+n}$, it is equivalent to $c\circ c=\bar{0}$. Note that ``$c\circ c=\bar{0} \Rightarrow \sum c_ic_{i+n}=\bar{0}$'' does not hold for $p=2$. Thus, for $p\neq2$, we can conclude that $\Lambda_p(\mathcal{C})$ is even with $g$ if and only if $c\circ c=\bar{0}$ for all $c\in\mathcal{C}$.
\end{proof}
\begin{prop} \label{prop:lambda_p_self-dual}
	For a prime $p$ (including $p=2$), $\Lambda_p(\mathcal{C})$ is self-dual with the metric $g$ if and only if $\mathcal{C}$ is self-dual with the metric $\bar{g}$.
\end{prop}
\begin{proof}
	First, we show that the dual lattice of the code is the lattice of the dual code, i.e. $(\Lambda_p(\mathcal{C}))^\ast=\Lambda_p(\mathcal{C}^\ast)$. We prove $(\Lambda_p(\mathcal{C}))^\ast \subset \Lambda_p(\mathcal{C}^\ast)$ explicitly. From the definition of the dual lattice, $x'\in (\Lambda_p(\mathcal{C}))^\ast$ when $x\cdot x'\in\mathbb{Z}$ for all $x\in\Lambda_p(\mathcal{C})$. From the construction \eqref{eq:lambda_p_C}, it can be rewritten that $\frac{1}{\sqrt{p}}(R(c)+pm)\cdot x' \in\mathbb{Z}\cdots(\star)$ for all $c\in\mathcal{C}$ and $m\in\mathbb{Z}^{2n}$. By considering the case $c=\vec{\bar{0}}$, only one element of $m$ is 1 and the others are 0, all elements of $x'$ should be multiples of $1/\sqrt{p}$. Since any integers can be expressed as $R(c')+pm'$, we can write $x'=\frac{1}{\sqrt{p}}(R(c')+pm')$ by $c'\in F_p^{2n}$ and $m'\in\mathbb{Z}^{2n}$, and $(\star)$ becomes $\frac{1}{p} (R(c)+pm)\cdot(R(c')+pm') \in\mathbb{Z}$. This is satisfied if and only if $R(c)\cdot R(c')\in p\mathbb{Z}$, which is equivalent to $c\circ c'=\bar{0}$. From the definition of the dual code, it means $c'\in\mathcal{C}^\ast$ and thus $x'\in\Lambda_p(\mathcal{C}^\ast)$. Since the discussion can be easily traced back to prove $(\Lambda_p(\mathcal{C}))^\ast \supset \Lambda_p(\mathcal{C}^\ast)$, we can conclude that $(\Lambda_p(\mathcal{C}))^\ast = \Lambda_p(\mathcal{C}^\ast)$.	Therefore, $\Lambda_p(\mathcal{C})$ is self-dual if and only if $\Lambda_p(\mathcal{C}) = \Lambda_p(\mathcal{C}^\ast)$, i.e. $\mathcal{C}$ is self-dual.
\end{proof}
\begin{prop} \label{prop:lambda_p_even_self-dual}
	For a prime $p\neq2$, $\Lambda_p(\mathcal{C})$ is even self-dual with the metric $g$ if and only if $\mathcal{C}$ is self-dual with the metric $\bar{g}$.
\end{prop}
\begin{proof}
	It is obvious from Propositions \ref{prop:lambda_p_even}, \ref{prop:lambda_p_self-dual} and the fact that $\mathcal{C}=\mathcal{C}^\ast$ leads $c\circ c=\bar{0}$ for all $c\in\mathcal{C}$.
\end{proof}
Since we have to handle the case $p=2$ differently as shown in the proof, we will only consider a prime $p\neq2$.

We have seen that we can generate even self-dual lattices from self-dual codes. Conversely, if a lattice $\Lambda\subset \mathbb{R}^{2n}$ is even self-dual with the metric $g$ and satisfies $\sqrt{p}\mathbb{Z}^{2n}\subset\Lambda$, there exists a code $\mathcal{C}\subset F_p^{2n}$ that satisfies $\Lambda=\Lambda_p(\mathcal{C})$. This is because $\Lambda\subset\frac{1}{\sqrt{p}}\mathbb{Z}^{2n}$ from the self-duality and $\Lambda$ can be written as $\Lambda=\{\frac{R(c)+pm}{\sqrt{p}}\mid c\in\mathcal{C},m\in\mathbb{Z}^{2n}\}$ from linearity of the lattice. Thus, by combining with Proposition \ref{prop:lambda_p_even_self-dual}, we can conclude that even self-dual lattices $\Lambda\subset\mathbb{R}^{2n}$ with the metric $g$ that satisfy $\sqrt{p}\mathbb{Z}^{2n}\subset\Lambda$ and self-dual codes $\mathcal{C}\subset F_p^{2n}$ with the metric $\bar{g}$ have a one-to-one correspondence.

Next, we want to find a generator matrix $G$ of a self-dual code $\mathcal{C}\subset F_p^{2n}$. From linearity, when $\mathcal{C}$ is $k$ dimensional, i.e. $|\mathcal{C}|=p^k$, $\mathcal{C}^\ast$ is $2n-k$ dimensional. Thus, $\mathcal{C}=\mathcal{C}^\ast$ requires $\mathcal{C}$ to be $n$ dimensional and $G$ to be a $2n\times n$ matrix. By the method using row reduction in Appendix \ref{app:generator_matrix}, we can find $G$ that can be written as
\begin{equation}
	G = \bar{S} \begin{pmatrix} \bar{I} \\ \bar{B} \end{pmatrix}
\end{equation}
where $\bar{B}$ is an $n \times n$ matrix on $F_p$ and $\bar{S}$ is a $2n\times 2n$ matrix that has one $\bar{1}$ and $2n-1$ $\bar{0}$'s in each row and satisfies $\bar{S}^T\bar{g}\bar{S}=\bar{g}$. $\bar{S}$ can be regarded as the operation to swap rows while keeping the inner product.

We want to find a condition for $\bar{B}$. From the definition of the dual code, $\mathcal{C}$ is self-dual only if $c\circ c'=\bar{0}$ for all $c,c'\in\mathcal{C}$. By the generator matrix, it can be rewritten as $a'^TG^T\bar{g}Ga=\bar{0}$ for all $a,a'\in F_p^{2n}$. This is equivalent to $G^T\bar{g}G=\bar{O}$ where $\bar{O}$ is a zero matrix, thus
\begin{equation}
	\bar{O} =
	\begin{pmatrix} \bar{I}&\bar{B}^T \end{pmatrix} \bar{S}^T \bar{g} \bar{S} \begin{pmatrix} \bar{I} \\ \bar{B} \end{pmatrix}
	= \begin{pmatrix} \bar{I}&\bar{B}^T \end{pmatrix} \bar{g} \begin{pmatrix} \bar{I} \\ \bar{B} \end{pmatrix} = \bar{B}+\bar{B}^T,
\end{equation}
i.e. $\bar{B}$ is antisymmetric in the sense of $F_p$.

Using these results, we can relate Narain CFTs and self-dual codes through even self-dual lattices as follows. This is one of the main purposes of this paper.
\begin{prop} \label{prop:code_CFT_rel}
	A self-dual code $\mathcal{C}\subset F_p^{2n}$ with the metric $\bar{g}$ has a generator matrix
	\begin{equation} \label{eq:self-dual_code_gen_mat}
		G = \bar{S} \begin{pmatrix} \bar{B} \\ \bar{I} \end{pmatrix}
	\end{equation}
	where $\bar{B}$ is an $n\times n$ matrix s.t. $\bar{B}+\bar{B}^T=\bar{O}$ and $\bar{S}$ is a $2n\times 2n$ matrix that has one $\bar{1}$ and $2n-1$ $\bar{0}$'s in each row and satisfies $\bar{S}^T\bar{g}\bar{S}=\bar{g}$. Then,
	\begin{equation} \label{eq:lat_rel_flat}
		\Lambda_p(\mathcal{C}) = S \Lambda_N(\sqrt{2/p},I,B')
	\end{equation}
	where $S=R(\bar{S})$ and $B'$ is an antisymmetric matrix s.t. $B'_{ij}=-B'_{ji}=R(\bar{B}_{ij})$ for $i\geq j$.
\end{prop}
\begin{proof}
	One generator matrix of $S\Lambda_N$ is
	\begin{equation}
		G_\Lambda = S \begin{pmatrix}
			\sqrt{p}I & \frac{1}{\sqrt{p}} B' \\ O & \frac{1}{\sqrt{p}}I
		\end{pmatrix}
	\end{equation}
	from \eqref{eq:lambda_N_gen_mat} and this is also the generator matrix of $\Lambda_p(\mathcal{C})$ since
	\begin{gather}
		\frac{1}{\sqrt{p}}R(c=Ga) = G_\Lambda \begin{pmatrix} \Box_{(n)} \\ R(a) \end{pmatrix}, \\
		S \begin{pmatrix} \vec{0}_{(i-1)} \\ \sqrt{p} \\ \vec{0}_{(n-i)} \\ \\ \vec{0}_{(n)} \\ \\ \end{pmatrix} = G_\Lambda \begin{pmatrix} \vec{0}_{(i-1)} \\ 1 \\ \vec{0}_{(n-i)} \\ \\ \vec{0}_{(n)} \\ \\ \end{pmatrix},\quad
		S \begin{pmatrix} \\ \vec{0}_{(n)} \\ \\ \vec{0}_{(i-1)} \\ \sqrt{p} \\ \vec{0}_{(n-i)} \end{pmatrix} = G_\Lambda \begin{pmatrix} \\ \Box_{(n)} \\ \\ \vec{0}_{(i-1)} \\ p \\ \vec{0}_{(n-i)} \end{pmatrix}
	\end{gather}
	where $\Box$'s mean appropriate vectors and we can generate any vector in $\Lambda_p(\mathcal{C})$ by combining these.
\end{proof}

Note that $S$ in \eqref{eq:lat_rel_flat} does not affect even self-duality of the lattice since $S^TgS=g$. By this proposition, several Narain CFTs can be constructed from a self-dual code depending on the choice of the generator matrix of the code. Since they generate the same lattice up to $S$ and have the same quantities as we will see, we simply call them the corresponding CFTs.

From \eqref{eq:lat_rel_flat}, the corresponding CFTs have dimensionless radii $r=\sqrt{2/p}$, the flat metric and integer backgrounds. Conversely, any Narain CFT with such properties can be constructed from a self-dual code up to $B_{ij}\sim B_{ij} +p$ in the background, which corresponds to T-duality.

\subsection{The spectral gap and the partition function} \label{sec:quantities}
Through Proposition \ref{prop:code_CFT_rel}, we can describe the important quantities in the corresponding Narain CFTs in the language of the code. From \eqref{eq:spectral_gap}, the spectral gap $\Delta$ is
\begin{equation}
\begin{aligned}
	\Delta &= \min_{\substack{x\in\Lambda_N(r,I,B) \\ x\neq0}} \frac{1}{2} x^Tx
	= \min_{\substack{y\in\Lambda_p(\mathcal{C}) \\ y\neq0}} \frac{1}{2} y^TSS^Ty 
	= \min_{\substack{y\in\Lambda_p(\mathcal{C}) \\ y\neq0}} \frac{1}{2} y^Ty \\
	&= \min_{\substack{c\in\mathcal{C},m\in\mathbb{Z}^{2n} \\ R(c)+pm\neq0}} \frac{1}{2} \sum_{i=1}^{2n} \left( \frac{R(c_i)+pm_i}{\sqrt{p}} \right)^2 \\
	&= \frac{1}{2p} \min \left\{ \min_{\substack{c\in\mathcal{C}, m\in\mathbb{Z}^{2n} \\ c\neq0}} \sum_{i=1}^{2n} (R(c_i)+pm_i)^2, \min_{\substack{m\in\mathbb{Z}^{2n} \\ m\neq0}} \sum_{i=1}^{2n} (pm_i)^2 \right\} \\
	&= \frac{1}{2p} \min \left\{ \min_{\substack{c\in\mathcal{C} \\ c\neq0}} \sum_{i=1}^{2n} \min\{R(c_i)^2, (R(c_i)-p)^2\}, p^2 \right\} \\
	&= \frac{1}{2p} \min \left\{ D(\mathcal{C})^2, p^2 \right\}. \label{eq:gap_from_code}
\end{aligned}
\end{equation}
Thus, under the condition $D(\mathcal{C})\leq p$, the spectral gap $\Delta$ represents the error correction capability of the code. Through this relation, for fixed $p,n$, searching for the code with the highest correction capability among self-dual codes on $F_p^{2n}$ is equivalent to searching for the theory with the largest spectral gap among Narain CFTs with $n$ scalars, dimensionless radii $r=\sqrt{2/p}$, the flat metric and integer backgrounds.

Figure \ref{fig:gap} shows the largest spectral gap of Narain CFTs constructed from self-dual codes on $F_p^{2n}$. We computed all the values by a full search using \eqref{eq:gap_from_code} for codes with generator matrices \eqref{eq:self-dual_code_gen_mat}. In (a), we varied $n$ at $p=3,5,7,11$ and calculated values evaluate `optimal' values from below, which are the largest spectral gaps among all Narain CFTs with $n$ scalars cited from Table 1 in \cite{2006.04839}. In (b),(c), we varied $p$ at $n=2,3$. The values at $n=2$ suggest that their upper limit is $1/\sqrt{3}$, which can be checked analytically by reducing to the sphere packing in two dimensions as follows.

For $n=2$, a generator matrix of a self-dual code $\mathcal{C}\subset F_p^{2n}$ is
\begin{equation}
	\bar{S} \begin{pmatrix} \bar{0}&\bar{a} \\ -\bar{a}&\bar{0} \\ \bar{1}&\bar{0} \\ \bar{0}&\bar{1} \end{pmatrix}.
\end{equation}
We can reduce it to a code $\mathcal{C}'\subset F_p^2$ generated by $(\bar{a},\bar{1})^T$ and
\begin{equation}
	\Delta = \min_{\substack{y\in\Lambda_p(\mathcal{C}) \\ y\neq0}} \frac{1}{2} y^Ty
	= \min_{\substack{y'\in\Lambda_p(\mathcal{C}') \\ y'\neq0}} \frac{1}{2} y'^Ty'.
\end{equation}
From the construction \eqref{eq:lambda_p_C}, $\Lambda_p(\mathcal{C}')\subset\mathbb{R}^2$ has $p$ points per $\sqrt{p}\times\sqrt{p}$. On the other hand, it is widely known that the highest density of the sphere packing in two dimensions is $\pi/(2\sqrt{3})$. Thus, the largest radius $r$ of circles centered at points in $\Lambda_p(\mathcal{C'})$ that do not intersect each other satisfies $\pi r^2 \leq \pi/(2\sqrt{3})$ and
\begin{equation}
	\Delta = \frac{1}{2} (2r)^2 \leq \frac{1}{\sqrt{3}}.
\end{equation}
For a general integer $n$, we do not know the analytical upper bound.

\begin{figure}[tbp]
	\begin{minipage}[b]{\textwidth}
		\centering
		\includegraphics[width=.45\textwidth]{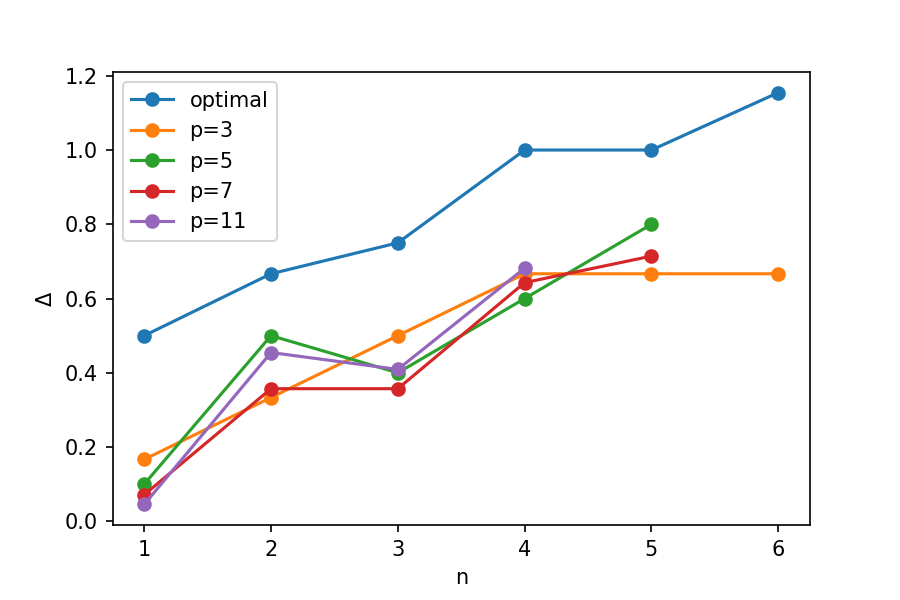}
		\subcaption{$p=3,5,7,11$}
	\end{minipage}
	\begin{minipage}[b]{0.45\textwidth}
		\centering
		\includegraphics[width=\textwidth]{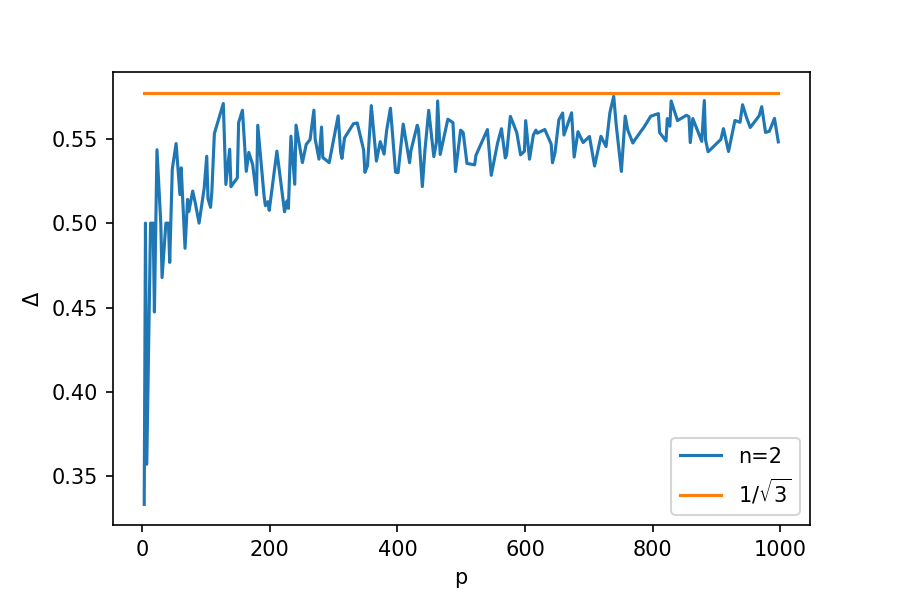}
		\subcaption{$n=2$}
	\end{minipage}
	\begin{minipage}[b]{0.45\textwidth}
		\centering
		\includegraphics[width=\textwidth]{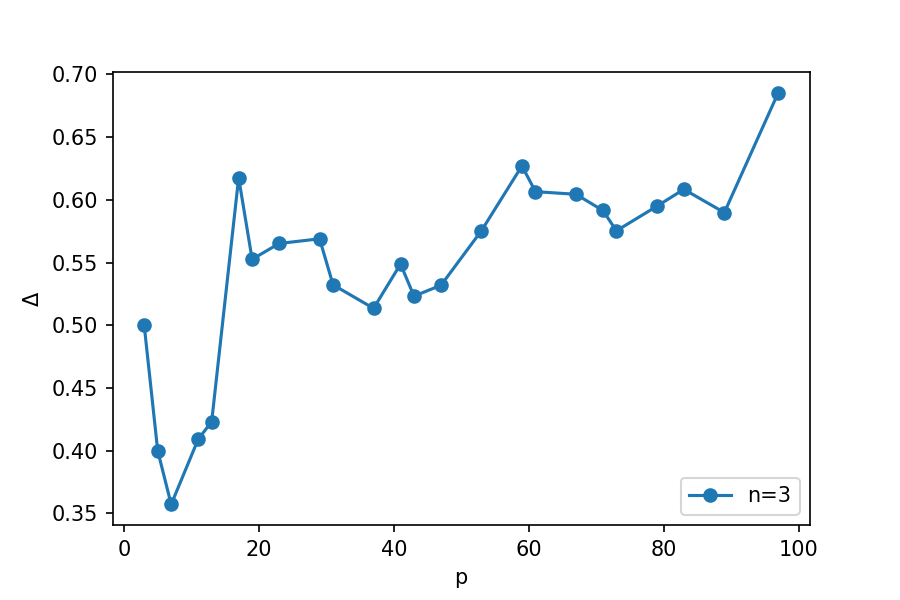}
		\subcaption{$n=3$}
	\end{minipage}
	\caption{The largest spectral gap of Narain CFTs constructed from self-dual codes on $F_p^{2n}$.}
	\label{fig:gap}
\end{figure}

Next, we consider the partition function. From \eqref{eq:partition_function},
\begin{equation}
\begin{aligned}
	&|\eta(\tau)|^{2n} Z(\tau) \\
	&= \sum_{x\in\Lambda_N(r,I,B)} q^{\sum_{i=1}^n (x_i+x_{i+n})^2/4} \bar{q}^{\sum_{i=1}^n (x_i-x_{i+n})^2/4} \\
	&= \sum_{y\in\Lambda_p(\mathcal{C})} q^{\sum_{i=1}^n (y_i+y_{i+n})^2/4} \bar{q}^{\sum_{i=1}^n (y_i-y_{i+n})^2/4} \\
	&= \sum_{c\in\mathcal{C}} \sum_{m\in\mathbb{Z}^{2n}} \prod_{i=1}^n q^{(R(c_i)+pm_i+R(c_{i+n})+pm_{i+n})^2/4p} \bar{q}^{(R(c_i)+pm_i-R(c_{i+n})-pm_{i+n})^2/4p} \\
	&= \sum_{c\in\mathcal{C}} \prod_{i=1}^n \sum_{m,l\in\mathbb{Z}} q^{(R(c_i)+R(c_{i+n})+p(m+l))^2/4p} \bar{q}^{(R(c_i)-R(c_{i+n})+p(m-l))^2/4p}.
\end{aligned}
\end{equation}
From the second line to the third line, we used the equation:
\begin{equation}
	\sum_{i=1}^n (x_i\pm x_{i+n})^2 = x^T(I\pm g)x = y^TS(I\pm g)S^Ty = y^T(I\pm g)y = \sum_{i=1}^n(y_i\pm y_{i+n})^2.
\end{equation}
Thus, the partition function $Z(\tau)$ is
\begin{equation}
	Z(\tau) = |\eta(\tau)|^{-2n} \sum_{c\in\mathcal{C}} \prod_{x,y\in F_p} (t_{x,y})^{w_{x,y}(c)}
\end{equation}
where
\begin{equation}
	t_{x,y} = \sum_{m,l\in\mathbb{Z}} q^{(R(x)+R(y)+p(m+l))^2/4p} \bar{q}^{(R(x)-R(y)+p(m-l))^2/4p}
\end{equation}
and $w_{x,y}(c)=|\{i\in\{1,\dots,n\}\mid(c_i,c_{i+n})=(x,y)\}|$. It has the form like the enumerator polynomial in coding theory. We can check modular invariance, i.e. $Z(\tau)=Z(\tau+1)$ and $Z(\tau)=Z(-1/\tau)$,  directly from this form.

\subsection{General construction} \label{sec:construction_general}
In section \ref{sec:construction_flat}, we constructed Narain CFTs with dimensionless radii $r=\sqrt{2/p}$, the flat metric and integer backgrounds from self-dual codes. By varying the correspondence between codes and lattices in \eqref{eq:lambda_p_C}, we can construct more general Narain CFTs.

We define a lattice corresponding to an additive code $\mathcal{C}\subset F_p^{2n}$ and a $2n\times 2n$ matrix $V=(v_1,\dots,v_{2n})\in GL(2n,\mathbb{R})$ by
\begin{equation}
	\Lambda_p(\mathcal{C};V) = \left\{ \sum_{i=1}^{2n} (R(c_i)+pm_i)v_i \,\middle|\, c\in\mathcal{C}, m\in\mathbb{Z}^{2n} \right\} \subset \mathbb{R}^{2n}.
\end{equation}

As in section \ref{sec:construction_flat}, we want to relate even self-duality in lattices and self-duality in codes. In fact, the following proposition holds.

\begin{prop} \label{prop:lambda_p_CV_even_self-dual}
	``For a primes $p\neq2$, $\Lambda_p(\mathcal{C};V)$ is even self-dual with the metric $g$ if and only if $\mathcal{C}$ is self-dual with the metric $\bar{g}$'' holds if conditions:
	\begin{itemize}
		\item $W:=pV^TgV$ is unimodular and $\forall i, W_{i,i}\in2\mathbb{Z}$
		\item $W \equiv R(\bar{g}) \mod p$
	\end{itemize}
	are satisfied.
\end{prop}

\begin{proof}
	(A proof for the more general case is in Appendix \ref{app:even_self-dual}.)
\end{proof}

Using this proposition, we can construct Narain CFTs from self-dual codes as follows.
\begin{prop} \label{prop:code_CFT_rel_general}
	A self-dual code $\mathcal{C}\subset F_p^{2n}$ with the metric $\bar{g}$ has a generator matrix
	\begin{equation}
		G = \bar{S} \begin{pmatrix} \bar{X} \\ \bar{I} \end{pmatrix}
	\end{equation}
	where $\bar{X}$ is an $n\times n$ matrix s.t. $\bar{X}+\bar{X}^T=\bar{O}$ and $\bar{S}$ is a $2n\times 2n$ matrix that has one $\bar{1}$ and $2n-1$ $\bar{0}$'s in each row and satisfies $\bar{S}^T\bar{g}\bar{S}=\bar{g}$. In addition, if an $n\times n$ matrix $\gamma$ and an $n\times n$ antisymmetric matrix $B$ satisfy
	\begin{equation} \label{eq:gen_mat_rel}
		VS \begin{pmatrix} pI&X' \\ 0&I \end{pmatrix} = S \begin{pmatrix}
			\frac{\sqrt{2}}{r}\gamma^{-1} & \frac{r}{\sqrt{2}}\gamma^{-1}B \\
			0 & \frac{r}{\sqrt{2}} \gamma^T
		\end{pmatrix}
	\end{equation}
	where $S=R(\bar{S})$ and $X'$ is an antisymmetric matrix s.t. $X'_{ij}=-X'_{ji}=R(\bar{X}_{ij})$ for $i\geq j$, then
	\begin{equation} \label{eq:lat_rel}
		\Lambda_p(\mathcal{C};V) = S \Lambda_N(r,\gamma\gamma^T,B).
	\end{equation}
\end{prop}
\begin{proof}
	As in Proposition \ref{prop:code_CFT_rel}, both sides of \eqref{eq:gen_mat_rel} are generator matrices of lattices in \eqref{eq:lat_rel}.
\end{proof}
Note again that $S$ in \eqref{eq:lat_rel} does not affect even self-duality of the lattice since $S^TgS=g$. This is a generalization of Proposition \ref{prop:code_CFT_rel}, which is the case $V=I_{2n}/\sqrt{p}$. Since the form of the corresponding Narain CFT depends strongly on $V$ and it is difficult to handle collectively, we will see some specific examples that may be significant. For simplicity, we only consider codes with $\bar{S}=\bar{I}$, i.e. $S=I$ in \eqref{eq:gen_mat_rel}, \eqref{eq:lat_rel}. The results can be easily extended to other cases.

First, we set
\begin{equation}
	V_a = \frac{1}{\sqrt{p}} \begin{pmatrix} aI&0 \\ 0&\frac{1}{a}I \end{pmatrix},\, a\in\mathbb{R},
\end{equation}
which satisfy $pV_a^TgV_a=g=R(\bar{g})$. A self-dual code $\mathcal{C}$ with a generator matrix $G=(\bar{X}^T,\bar{I})^T$ satisfies
\begin{equation}
	\Lambda_p(\mathcal{C};V_a) = \Lambda_N(\sqrt{2/p}/a,I,a^2X').
\end{equation}
The spectral gap of the corresponding Narain CFT can be written as
\begin{equation}
	\Delta = \frac{1}{2p} \min\{D_a(\mathcal{C})^2, a^2p^2, \frac{1}{a^2}p^2 \}
\end{equation}
where
\begin{gather}
	D_a(\mathcal{C}) = \min_{c\in\mathcal{C},c\neq0} d_a(c), \\
	d_a(c) = \sqrt{ a^2 \sum_{i=1}^{n} ( \min \{ R(c_i), R(-c_i)\})^2 + \frac{1}{a^2} \sum_{i=n+1}^{2n} (\min \{ R(c_i), R(-c_i)\})^2 } .
\end{gather}
By introducing the factor $a$ in $V$, we multiply radii and background by $1/a$ and $a^2$ in the CFT. The calculation for the spectral gap is almost same as \eqref{eq:gap_from_code}. $D_a(\mathcal{C})$ can be regarded as the error correction capability of the code when the probability of error is proportional to $a$ in the former half ($i=1,\dots,n$) and $1/a$ in the latter half ($i=n+1,\dots,2n$).

Next, we set
\begin{equation}
	V(\alpha,\beta) = \frac{1}{\sqrt{p}} \begin{pmatrix} \alpha^{-1}&\alpha^{-1}\beta \\ 0&\alpha^T \end{pmatrix}
\end{equation}
where $\alpha$ is an $n\times n$ matrix and $\beta$ is an $n\times n$ antisymmetric matrix, which satisfy $pV(\alpha,\beta)^TgV(\alpha,\beta)=g=R(\bar{g})$. A self-dual code $\mathcal{C}$ with a generator matrix $G=(\bar{X}^T,\bar{I})^T$ satisfies
\begin{equation}
	\Lambda_p(\mathcal{C};V(\alpha,\beta)) = \Lambda_N(\sqrt{2/p},\alpha\alpha^T,X'+\beta).
\end{equation}
We can relate the code to any metric and background by varying $\alpha$ and $\beta$. However, the quantities in CFT such as the spectral gap cannot be expressed in a simple form in the language of the code and it is difficult to interpret in the context of coding theory. Thus, we should restrict our consideration to simple forms such as $I_{2n}/\sqrt{p}$ and $V_a$ in order to see ``good'' relation between error-correcting codes and Narain CFTs. 

\section{Conclusion}
We defined the correspondence between Narain CFTs with $n$ scalars and self-dual error-correcting codes on $F_p^{2n}$ through even self-dual lattices on $\mathbb{R}^{n,n}$, and derived the relation between quantities in both theories. We found that the spectral gap of the CFT and the error correction capability of the code are almost proportional and the partition function of the CFT can be expressed by the extended enumerator polynomial of the code. Using this relation, we calculated the specific spectral gaps to get the maximum value among the Narain CFTs with radii $\sqrt{2/p}$, the flat metric and integer backgrounds and compared them with the known values among all Narain CFTs.

In general, a number of elements in a finite field is a prime number or a prime power. In this paper, we dealt with prime numbers and only showed the construction in Appendix \ref{app:even_self-dual} for prime powers. We will aim to derive meaningful results for CFTs from codes on finite fields with prime power elements.

In the future, we expect to analytically impose stronger bounds on quantities in both theories such as the spectral gap than those are currently known by considering the limit $p\to\infty$. In addition, we will clearly relate symmetries in Narain CFTs to counterparts in codes and see how they appear when we vary $V$ in Proposition \ref{prop:code_CFT_rel_general}.

\acknowledgments
I am grateful to Yutaka Matsuo, Kantaro Ohmori, Kohki Kawabata, Go Noshita and Shu Shimamura for useful discussions. This research was supported by FoPM, WINGS Program, the University of Tokyo.

\appendix
\section{A proof for generator matrices of self-dual codes}
\label{app:generator_matrix}
We will prove that a self-dual code $\mathcal{C}\subset F_p^{2n}$ with the metric $\bar{g}=\begin{pmatrix} \bar{0}&\bar{I}\\ \bar{I}&\bar{0} \end{pmatrix}$ has a generator matrix $G$ that can be written as
\begin{equation}
	G = \bar{S} \begin{pmatrix} \bar{I} \\ \bar{B} \end{pmatrix}
\end{equation}
where $\bar{B}$ is an $n \times n$ matrix and $\bar{S}$ is a $2n\times 2n$ matrix that has one $\bar{1}$ and $2n-1$ $\bar{0}$'s in each row and satisfies $\bar{S}^T\bar{g}\bar{S}=\bar{g}$.

Let $G'$ be any generator matrix of $\mathcal{C}$. Note that the matrix corresponds to the same code even if we multiply any column by any value without $\bar{0}$, add one column to another column and swap columns. In addition, $\bar{S}$ corresponds to swapping rows as $i$-th $\leftrightarrow$ $(i+n)$-th or $i,(i+n)$-th $\leftrightarrow$ $j,(j+n)$-th for generator matrices. We can basically get $G$ from $G'$ by the following method.
	\begin{itemize}
		\item[1-1] (1st column) $\rightarrow$ (1st column)$/G'_{11}$
		\item[1-2] For $j\in\{1,\dots,n\}\setminus\{1\}$, ($j$-th column) $\rightarrow$ ($j$-th column) $-$ (1st column)$\times G'_{1j}$ 
		\item[2-1] (2nd column) $\rightarrow$ (2nd column)$/G'_{22}$
		\item[2-2] For $j\in\{1,\dots,n\}\setminus\{2\}$, ($j$-th column) $\rightarrow$ ($j$-th column) $-$ (2nd column)$\times G'_{2j}$
		\item[\vdots]
		\item[n-1] ($n$-th column) $\rightarrow$ ($n$-th column)$/G'_{nn}$
		\item[n-2] For $j\in\{1,\dots,n\}\setminus\{n\}$, ($j$-th column) $\rightarrow$ ($j$-th column) $-$ ($n$-th column)$\times G'_{nj}$
	\end{itemize}
All calculations are understood as on $F_p$. When the step i-2 is done, the top $i$ rows must be $(\bar{I}_{i\times i}\;\bar{O}_{i\times (n-i)})$. A problem is that it doesn't work when $G'_{ii}=\bar{0}$. If there exists a nonzero $G'_{jk}$ or $G'_{j+n,k}$ for $i\leq j,k\leq n$, then we can bring it to $G'_{ii}$ by swapping rows by $\bar{S}$ and/or columns and restart the method. If not, i.e. $G'$ has the form (for simplicity, $i-1\to i$)
	\begin{equation}
		\begin{pmatrix}
			\bar{I}_{i\times i} & \bar{O}_{i\times (n-i)} \\
			\Box_{(n-i)\times i} & \bar{O}_{(n-i)\times (n-i)} \\
			\Box_{i\times i} & \Box_{i\times (n-i)} \\
			\Box_{(n-i)\times i} & \bar{O}_{(n-i)\times (n-i)} \\
		\end{pmatrix}
	\end{equation}
	where $\Box$'s mean arbitrary matrices, we can't complete the method. However, it can't happen when the code $\mathcal{C}$ is self-dual. If $G'$ has this form,
	\begin{equation}
		(\underbrace{\bar{0},\dots,\bar{0}}_{n}, -G'_{n,1}, \cdots, -G'_{n,i}, \underbrace{\bar{0},\dots,\bar{0}}_{n-i-1}, \bar{1} )^T
	\end{equation}
	is orthogonal with the metric $\bar{g}$ to all columns of $G'$ but can't written as $G'a$ by $a\in F_p^n$. It means $\mathcal{C}^\ast \not\subset \mathcal{C}$, which contradicts $\mathcal{C}=\mathcal{C}^\ast$. Thus, we can always complete the method and get the desired generator matrix $G$.

\section{A proof of Proposition \ref{prop:lambda_p_CV_even_self-dual}}
\label{app:even_self-dual}
As mentioned in section \ref{sec:construction_general}, we will prove the more general case. 

As preparation, we introduce a finite field with prime power elements. For a prime $p$ and an integer $l\in\mathbb{N}$, $F_{p^l}=F_p[x]/(f_{p,l}(x))=\{\sum_{t=0}^{l-1} a_tx^t \mid a_t\in F_p \}$ is a finite field with $p^l$ elements where $f_{p,l}(x)$ is the Conway polynomial and $(f(x))$ is the ideal generated by $f(x)$. Furthermore, for $r\in\mathbb{Z}$ s.t. $0\leq r<l$, we define a map $P_r: F_{p^l}\to F_p$ by $P_r(\sum_{t=0}^{l-1} a_tx^t)=a_r$.

We define a lattice corresponding to an additive code $\mathcal{C}\subset F_{p^l}^{n}$ and an $nl\times nl$ matrix $V=(v_{1,0},v_{1,1},\dots, v_{n,l-1}) \in GL(nl,\mathbb{R})$ by
\begin{equation}\label{eq:lambda_pl_CV}
	\Lambda_{p^l}(\mathcal{C};V) = \left\{ \sum_{i=1}^n \sum_{t=0}^{l-1} (R(c_{i,t})+pm_{i,t})v_{i,t} \,\middle|\, c\in\mathcal{C}, m\in\mathbb{Z}^{nl} \right\} \subset \mathbb{R}^{nl}
\end{equation}
where $c_i=\sum_{t=0}^{l-1}c_{i,t}x^t, c_{i,t}\in F_p$.

A generalization of Proposition \ref{prop:lambda_p_CV_even_self-dual} is as follows.

\begin{prop}
	We define inner products on $\mathbb{R}^{nl}$ and $F_{p^l}^n$ by $\mathbb{R}^{nl}\times\mathbb{R}^{nl}\to\mathbb{R}$, $x\cdot y=x^Tgy$ where $g: nl\times nl$ symmetric matrix and $F_{p^l}^n \times F_{p^l}^n \to F_p$, $a \circ b = \sum_{r=0}^{l-1} e_r P_{r}( a^T \bar{g} b)$ where $e_r\in F_p$ and $\bar{g}: n\times n$ symmetric matrix on $F_p$. Then,
	\begin{equation}
		\Lambda_{p^l}(\mathcal{C};V) \text{ is self-dual with the metric } g
		\;\Leftrightarrow\; \mathcal{C} \text{ is self-dual with the inner product } \circ
	\end{equation}
	holds if conditions:
	\begin{itemize}
		\item $W:=pV^TgV$ is unimodular
		\item $w_{i,t,j,s} := pv_{i,t}\cdot v_{j,s} (= W_{l(i-1)+t+1,l(j-1)+s+1}) \equiv R(\bar{g}_{ij}) R(\sum e_r P_r(x^{t+s})) \mod p$
	\end{itemize}
	are satisfied. Furthermore, for $p\neq2$,
	\begin{equation}
		\Lambda_{p^l}(\mathcal{C};V) \text{ is even with the metric $g$}
		\;\Leftrightarrow\; \forall c\in\mathcal{C},\; c\circ c = \bar{0}
	\end{equation}
	holds if above two conditions and an additional condition:
	\begin{itemize}
		\item $\forall i,t,\; w_{i,t,i,t} \in 2\mathbb{Z}$
	\end{itemize}
	are satisfied. For $p=2$, without additional conditions,
	\begin{equation}
		\Lambda_{p^l}(\mathcal{C};V) \text{ is even with the metric $g$}
		\;\Leftrightarrow\; \forall c\in\mathcal{C}, \sum_{i,t,j,s} R(c_{i,t})R(c_{j,s})w_{i,t,j,s} \in 4\mathbb{Z}.
	\end{equation}	
\end{prop}

\begin{proof}
	\textbf{(self-dual)} First, we show that the dual lattice of the code is the lattice of the dual code with the same $V$, i.e. $(\Lambda_{p^l}(\mathcal{C};V))^\ast=\Lambda_{p^l}(\mathcal{C}^\ast;V)$. We prove $(\Lambda_{p^l}(\mathcal{C};V))^\ast\subset\Lambda_{p^l}(\mathcal{C}^\ast;V)$ explicitly. From the definition of the dual lattice, $x'\in(\Lambda_{p^l}(\mathcal{C};V))^\ast$ when $x\cdot x'\in\mathbb{Z}$ for all $x\in\Lambda_{p^l}(\mathcal{C};V)$. Since $v_{i,t}$ are basis of $\mathbb{R}^{nl}$, we can write $x'=\sum_{j,s}x'_{j,s}v_{j,s}$ by $x'_{j,s}\in\mathbb{R}$. Then, from the construction \eqref{eq:lambda_pl_CV}, $x\cdot x'\in\mathbb{Z}$ can be rewritten that
	\begin{equation} \label{eq:self_dual_cond}
		\sum_{i,t,j,s}(R(c_{i,t})+pm_{i,t})v_{i,t}\cdot x'_{j,s}v_{j,s}
		= \sum_{i,t,j,s} \frac{1}{p} (R(c_{i,t})+pm_{i,t})x'_{j,s}w_{i,t,j,s}
	\end{equation}
	is an integer for all $c\in\mathcal{C}$ and $m\in\mathbb{Z}^{nl}$. By considering the case $c=\vec{\bar{0}}$, $x'$ should satisfy $\sum_{i,t,j,s} m_{i,t}x'_{j,s}w_{i,t,j,s} \in\mathbb{Z}$ for all $m\in\mathbb{Z}^{nl}$, which is equivalent to $\sum_{j,s} x'_{j,s}w_{i,t,j,s} \in \mathbb{Z}$ for all $i\in\{1,\dots,n\},t\in\{0,\dots,l-1\}$. Using the first condition that $W$ is unimodular, we get $x'_{j,s} \in \mathbb{Z}$. Since any integers can be expressed as $R(c')+pm'$, we can write $x'= \sum_{j,s} (R(c'_{j,s})+pm'_{j,s})v_{j,s}$ by $c'\in F_{p^l}^n$ and $m'\in\mathbb{Z}^{nl}$, and \eqref{eq:self_dual_cond} becomes that
	\begin{equation}
	\begin{aligned}
		&\frac{1}{p} \sum_{i,t,j,s} (R(c_{i,t})+pm_{j,s})(R(c'_{j,s})+pm'_{j,s})w_{i,t,j,s} \\
		=\; & \sum_{i,t,j,s} \left( \frac{1}{p}R(c_{i,t})R(c'_{j,s})+R(c_{i,t})m'_{j,s}+m_{i,t}R(c'_{j,s}) + pm_{i,t}m'_{j,s} \right) w_{i,t,j,s}
	\end{aligned}
	\end{equation}
	is an integer. The second, third and fourth terms are always integers, thus this is satisfied if and only if
	\begin{equation}
		\sum_{i,t,j,s} R(c_{i,t})R(c'_{j,s})w_{i,t,j,s} \in p\mathbb{Z}.
	\end{equation}
	Using the second condition that $w_{i,t,j,s} \equiv R(\bar{g}_{ij}) R(\sum e_r P_r(x^{t+s})) \mod p$, we get
	\begin{equation}
		\sum_{i,t,j,s} R(c_{i,t})R(c'_{j,s})R(\bar{g}_{ij})R(\sum_r e_rP_r(x^{t+s})) \in p\mathbb{Z}.
	\end{equation}
	In the language of $F_{p}$, it can be rewritten as
	\begin{equation}
		\sum_{i,t,j,s,r} e_rP_r(c_{i,t}x^tc'_{j,s}x^s\bar{g}_{ij}) = \sum_{i,j,r} e_rP_r(c_ic'_j\bar{g}_{ij}) = c\circ c' = \bar{0}.
	\end{equation}
	From the definition of the dual code, it means $c'\in\mathcal{C}^\ast$ and thus $x'\in\Lambda_{p^l}(\mathcal{C}^\ast;V)$, which is what we wanted to show. Since the discussion can be easily traced back to prove $(\Lambda_{p^l}(\mathcal{C};V))^\ast\supset\Lambda_{p^l}(\mathcal{C}^\ast;V)$, we can conclude that $(\Lambda_{p^l}(\mathcal{C};V))^\ast=\Lambda_{p^l}(\mathcal{C}^\ast;V)$. Therefore, $\Lambda_{p^l}(\mathcal{C};V)$ is self-dual if and only if $\Lambda_{p^l}(\mathcal{C};V)=\Lambda_{p^l}(\mathcal{C}^\ast;V)$, i.e. $\mathcal{C}$ is self-dual with the inner product $\circ$.

	\textbf{(even)}
	From the definition, $\Lambda_{p^l}(\mathcal{C};V)$ is even with $g$ when $x\cdot x\in2\mathbb{Z}$ for all $x\in\Lambda_{p^l}(\mathcal{C};V)$. From the construction \eqref{eq:lambda_pl_CV}, it can be rewritten that
	\begin{equation}
	\begin{aligned}
		&\sum_{i,t,j,s} (R(c_{i,t})+pm_{i,t})v_{i,t} \cdot (R(c_{j,s})+pm_{j,s})v_{j,s} \\
		= &\sum_{i,t,j,s} \left( \frac{1}{p} R(c_{i,t})R(c_{j,s}) + R(c_{i,t})m_{j,s} + m_{i,t}R(c_{j,s}) + p m_{i,t}m_{j,s} \right) w_{i,t,j,s}
	\end{aligned}
	\end{equation}
	is even. The sum of the second and third terms is always even from $w_{i,t,j,s}=w_{j,s,i,t}\in\mathbb{Z}$, thus it can be divided into
	\begin{equation}\label{eq:for_c}
		\frac{1}{p} \sum_{i,t,j,s} R(c_{i,t})R(c_{j,s})w_{i,t,j,s} \in 2\mathbb{Z}
	\end{equation}
	for all $c\in\mathcal{C}$ and
	\begin{equation}\label{eq:for_m}
		p \sum_{i,t,j,s} m_{i,t}m_{j,s}w_{i,t,j,s} \in 2\mathbb{Z}
	\end{equation}
	for all $m\in\mathbb{Z}^{nl}$.

	For $p\neq2$, we can ignore $p$ in \eqref{eq:for_m} and
	\begin{equation}
		\sum_{i,t,j,s} m_{i,t}m_{j,s}w_{i,t,j,s}
		= \sum_{(i,t)\neq(j,s)} m_{i,t}m_{j,s}w_{i,t,j,s} + \sum_{i,t} m_{i,t}^2 w_{i,t,i,t}.
	\end{equation}
	Since $w_{i,t,j,s}=w_{j,s,i,t}\in\mathbb{Z}$, the first term of the right-hand side is always even and thus \eqref{eq:for_m} is equivalent to $\sum_{i,t}m_{i,t}^2w_{i,t,i,t} \in 2\mathbb{Z}$. It is satisfied for all $m\in\mathbb{Z}^{nl}$ if and only if $w_{i,t,i,t}\in2\mathbb{Z}$ for all $i,t$. In that case, $\sum R(c_{i,t})R(c_{j,s})w_{i,t,j,s}$ is always even. Thus, \eqref{eq:for_c} becomes $\sum_{i,t,j,s} R(c_{i,t})R(c_{j,s})w_{i,t,j,s} \in p\mathbb{Z}$, which is equivalent to $c \circ c = \bar{0}$ from the discussion in the proof for ``self-dual''. Combining these results, we get an equivalence relation:
	\begin{equation}
		\Lambda_{p^l}(\mathcal{C};V) \text{ is even}
		\;\Leftrightarrow\; \forall i,t,\; w_{i,t,i,t} \in 2\mathbb{Z} \;\text{and}\; \forall c\in\mathcal{C},\; c \circ c = \bar{0} \quad (p\neq2).
	\end{equation}

	For $p=2$, \eqref{eq:for_m} is automatically satisfied. Thus, from \eqref{eq:for_c},
	\begin{equation}
		\Lambda_{p^l}(\mathcal{C};V) \text{ is even}
		\;\Leftrightarrow\; \forall c\in\mathcal{C}, \sum_{i,t,j,s} R(c_{i,t})R(c_{j,s})w_{i,t,j,s} \in 4\mathbb{Z} \quad (p=2).
	\end{equation}
\end{proof}

Construction A in \cite{2107.02816} is essentially equivalent to the case $p=2,l=2$,
\begin{gather}
	g = S \otimes \begin{pmatrix} 0&1\\ 1&0 \end{pmatrix},\quad
	R(\bar{g}) = S,\quad
	e_0 = \bar{0},\quad
	e_1 = \bar{1}, \\
	\left\{ \begin{aligned}
		&v_{i,0} = \frac{1}{3^{1/4}}(0,\dots,0,\underset{\text{i-th}}{1},0,\dots,0) \otimes (1,0) \\
		&v_{i,1} = \frac{1}{3^{1/4}}(0,\dots,0,\underset{\text{i-th}}{1},0,\dots,0) \otimes \left(-\frac{1}{2},\frac{\sqrt{3}}{2}\right)
	\end{aligned} \right. .
\end{gather}

\end{document}